\documentclass[oribibl]{llncs}
\usepackage{amsmath,amssymb}
\usepackage{graphicx, hyperref}
\usepackage{cite}
\usepackage{wrapfig}

\let\doendproof\endproof
\renewcommand\endproof{~\hfill\qed\doendproof}

\newcommand{\NP}{\mathsf{NP}}
\renewcommand{\P}{\mathsf{P}}

\DeclareMathOperator{\topcross}{cr}

\DeclareMathOperator{\pagecross}{cr}
\DeclareMathOperator{\pagecrossed}{cre}

\usepackage{todonotes}

\title{Fixed parameter tractability of crossing minimization of almost-trees}
\author{Michael J. Bannister \and David Eppstein \and Joseph A. Simons}
\institute{Department of Computer Science, University of California, Irvine}

\begin{document}
\maketitle
\begin{abstract}
We investigate exact crossing minimization for graphs that differ from trees by a small number of additional edges, for several variants of the crossing minimization problem. In particular, we provide fixed parameter tractable algorithms for the $1$-page book crossing number, the $2$-page book crossing number, and the minimum number of crossed edges in $1$-page and $2$-page book drawings.
\end{abstract}

\pagestyle{plain}

\section{Introduction}
Graphs that differ from a tree by the inclusion of a small number of edges arise in many applications; such graphs are called \emph{almost-trees}. Almost-trees can be found in the areas of biology, medicine, operations research, sociology, genealogy, distributed systems, and telecommunications, and in each of these applications it is important to find effective visualizations\footnote{We provide more details about these applications in Section~\ref{sec:apps}.}. One of the most important criteria for the aesthetics and readability of a graph drawing is its number of its crossings. Although crossing minimization problems tend to be $\NP$-complete, we may hope that the graphs arising in applications are not hard instances for these problems, allowing us to find optimal drawings for them efficiently. In this paper we prove that almost-trees are indeed not hard instances by designing algorithms for crossing minimization of almost-trees that are fixed-parameter tractable when parameterized by the number of extra non-tree edges in these graphs.

Many different variants of the crossing number have been studied, depending on what types of drawing are allowed and what we count as a crossing~\cite{PacTot-JCTB-00}. The most frequently studied is the topological crossing number, $\topcross(G)$, which counts the number of crossings in a unrestricted placement of vertices and edges in the plane. In this paper we consider also the $1$-page and $2$-page crossing numbers, denoted $\pagecross_1(G)$ and $\pagecross_2(G)$ respectively. The $1$-page crossing number counts the minimal number of crossings in a drawing where all the vertices of $G$ are placed on a straight line, and all edges must be placed to one side of the line. The $2$-page crossing number is defined similarly: all vertices of $G$ are placed on a straight line and edges may be assigned to either side of the line, but are not allowed to cross the line. In both $1$-page and $2$-page drawings it is not uncommon to place the vertices on a circle instead of a straight line; this does not change the crossing structure of the drawing. In addition to the number of crossings we consider the number of crossed edges for these drawing styles, denoted $\pagecrossed_1(G)$ and $\pagecrossed_2(G)$. 
\begin{figure}[t]
\centering
\includegraphics[height=1in]{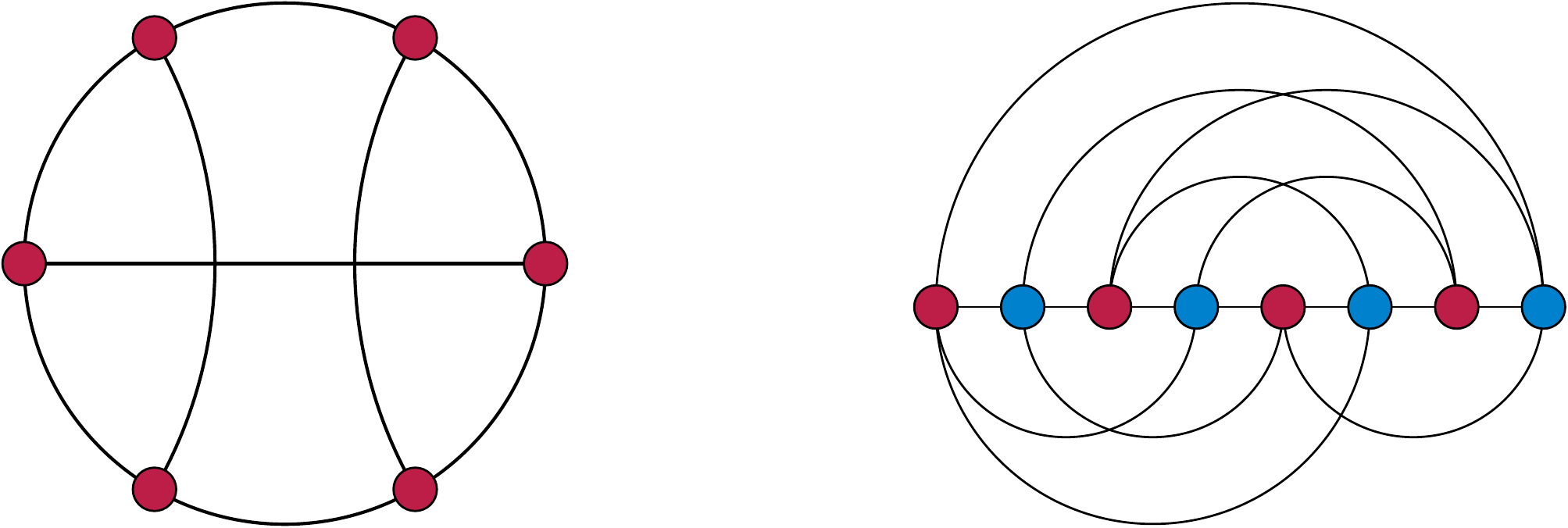}
\caption{Left: 1-page circular embedding with 
    two crossings and three crossed edges. 
    Right: 2-page linear embedding of $K_{4,4}$ with 
    four crossings and eight crossed edges.
\label{fig:1p2p}
}
\end{figure}

Following Gurevich, Stockmeyer and Vishkin \cite{GurStoVis-JACM-84} we define a \emph{$k$-almost-tree} to be a graph such that every biconnected component of the graph has cyclomatic number at most $k$, where the \emph{cyclomatic number} is the difference between the number of edges in a graph and in one of its maximal spanning forests. The $k$-almost-tree parameter has been used in past fixed-parameter algorithms~\cite{Fer-TSE-89, KloBodMul-ESA-93, AkuHayChi-JTB-07, JirKloKra-DAM-01, Cop-DAM-85, Bod-GTCCS-94, Fur-FUN-12}, and will play the same role in our algorithms for crossing minimization. 

Grohe and later Kawarabayashi and Reed showed the topological crossing number to be fixed parameter tractable for its natural parameter~\cite{Gro-STOC-01,KawRee-STOC-07}; the same is true for odd crossing number~\cite{PelSchSte-GD-07}. Because the topological crossing number is at most quadratic in the $k$-almost-tree parameter, $cr(G)$ is also fixed parameter tractable for $k$-almost-trees. However, to our knowledge no fixed parameter tractable algorithms were known for computing $1$-page or $2$-page crossing numbers. Indeed, for the $2$-page problem, determining whether a graph can be drawn with zero crossings is already $\NP$-complete~\cite{ChuLeiRos-SJADM-87}, so to achieve fixed parameter tractability we must use some other parameter such as the $k$-almost-tree parameter rather than using the crossing number itself as a parameter.

Our main results are that $\pagecross_1(G)$, $\pagecross_2(G)$, $\pagecrossed_1(G)$, and $\pagecrossed_2(G)$ are all fixed-parameter tractable for almost-trees.
As with previous work on parameterized algorithms for crossing numbers~\cite{Gro-STOC-01,KawRee-STOC-07,PelSchSte-GD-07}, our algorithms have a high dependence on their parameters. Making our algorithms more practical by reducing this dependence  remains an open problem.

\section{Application Domains}\label{sec:apps}
Examples of $k$-almost-trees can be found in biological gene expression networks, where vertices represent genes and edges represent correlations between pairs of genes. The $k$-almost-tree structure of such graphs has been exploited in parameterized algorithms for finding sequences of valid labelings of genes as active or inactive~\cite{AkuHayChi-JTB-07}. The parameter $k$ has also been used in algorithms for continuous facility location where weighted edges represent a road network on which to efficiently place facilities serving clients~\cite{GurStoVis-JACM-84}.  Intraprogram communication networks whose vertices represent modules of a distributed system and whose edges represent communicating pairs of modules also have an almost-tree structure that has been exploited for parameterized algorithms~\cite{Fer-TSE-89}.

A typical example of almost-trees arises when studying the spread of sexually transmitted infections, where sexual networks are constructed by voluntary survey. In these graphs vertices represent people who have received treatment, and edges represent their reported sexual parters. Analysis of these networks allows researchers to identify the growth and decline phases of an outbreak, and the general spread of the disease~\cite{PotPhiPlu-STI-02, DeSinWon-STI-04, PotMutRot-STI-02}. 

Another type of social network represents the business dealings of individuals and business entities. Examples of these networks can be found in the art of Mark Lombardi, an artist famous for his drawings of networks connecting the key players of conspiracy theories~\cite{LomHob-03}. Many of Lombardi's networks show an almost-tree structure; the Lippo Group Shipping network listed below is one.

The directed acyclic graphs originating from genealogical data where edges
represent parental relationships on the vertices are another example of
$k$-almost-trees, when viewed as undirected graphs, since in modern societies
marriage between close relatives is rare. Similar types of graphs also come from
animal pedigrees, academic family trees, and organizational lineages~\cite{ButActMar-JSS-08}.

Utility networks such as telecommunication networks and power grids also form an almost-tree structure, where additional edges beyond those of a spanning tree provide load balancing and redundancy. Since such links are expensive they are placed in the network sparingly.

\begin{figure}[b!]
    \centering
    \includegraphics[width=2in]{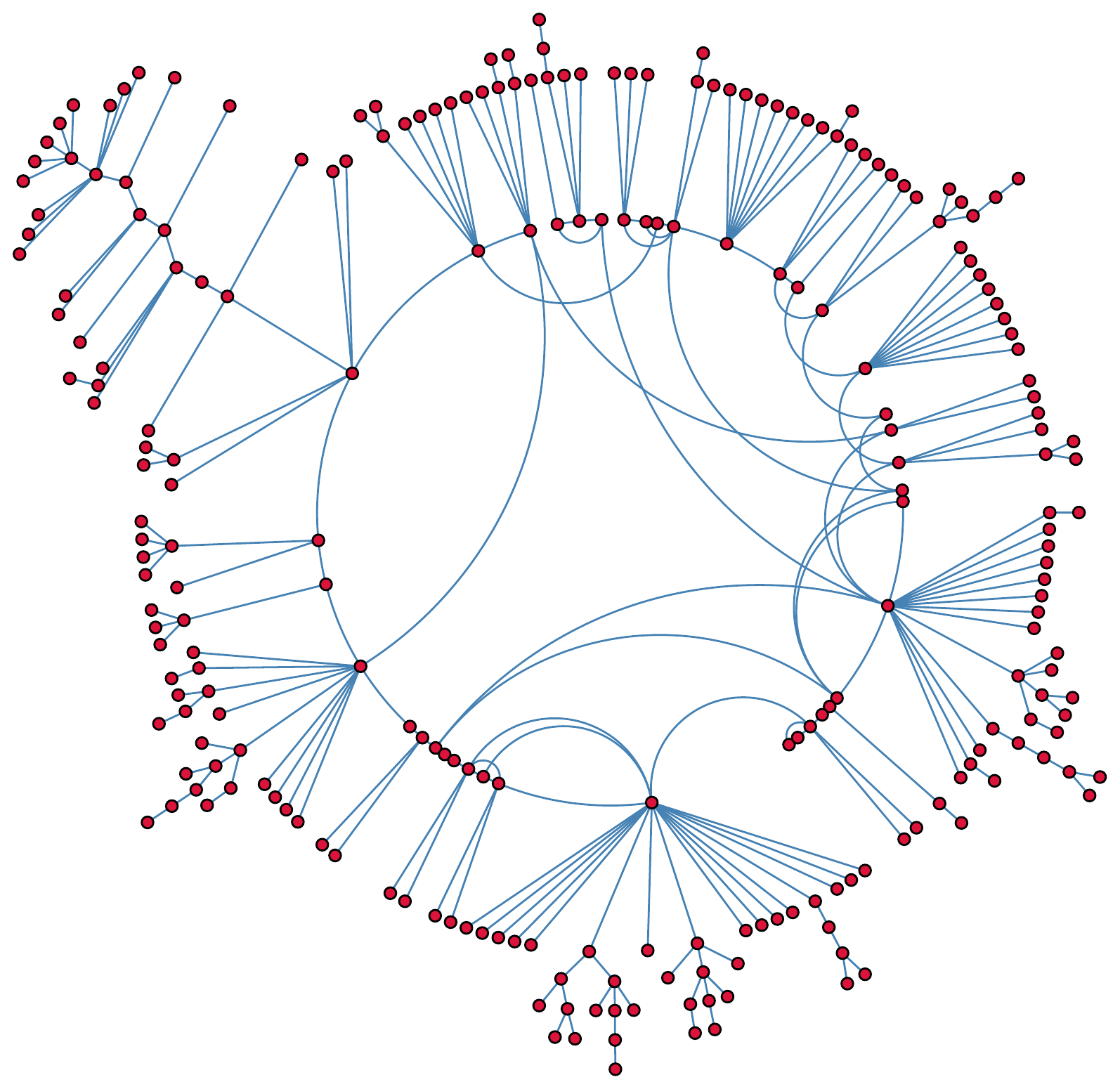}\hfill
    \includegraphics[width=2in]{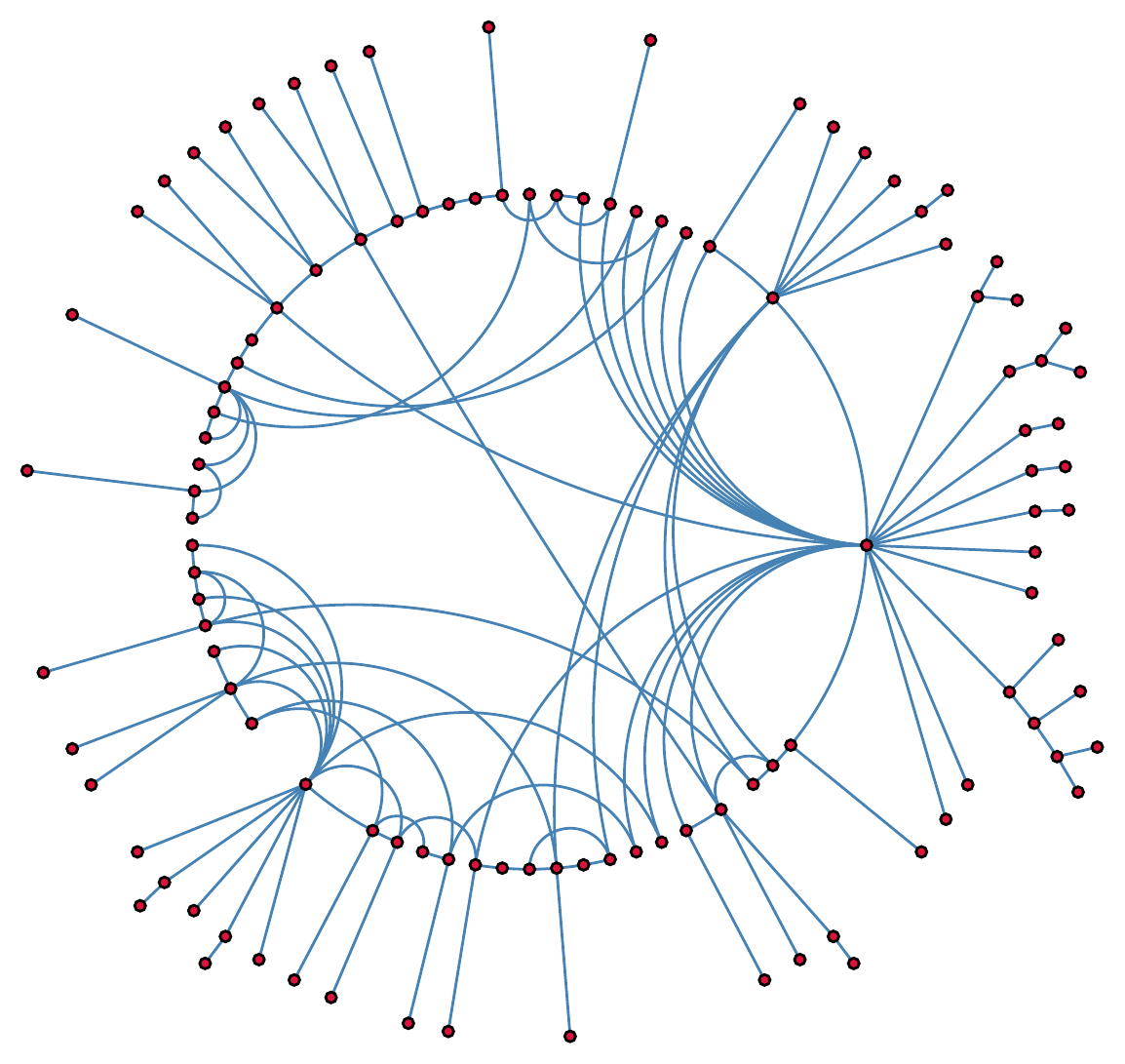}
    \caption{Two sunburst drawings. Left: An HIV infection graph. Right: Lombardi's World Finance Miami graph.}
    \label{fig:sunburst}
\end{figure}

In order to visually distinguish the tree-like parts of these graphs from the parts with nontrivial connectivity, we may use a \emph{sunburst} style (Figure~\ref{fig:sunburst}) in which the \emph{$2$-core} of the graph (the part of a graph which is left after repeatedly removing all degree one vertices~\cite{Sei-SN-83}) is drawn with a one-page circular layout and the rest of the graph extends outwards using a radial layout on concentric circles. In this style, crossings occur only within the inner one-page layout, motivating our interest in crossing minimization for one-page drawings of almost-trees. 

We collect statistics for several real world graphs in Table~\ref{tab:stats}. The table shows vertex and edge counts ($n$ and $m$), the cyclomatic number $a = m - n + 1$, the $k$-almost-tree parameter~$k$, and the vertex and edge counts for the $2$-core ($n_2$ and $m_2$). For most of these graphs the parameters $a$ and $k$ are low.

\vspace{-0.1in}
\begin{table}
\centering
\begin{tabular}{lcccccc}
Name& $n$& $m$& $a$& $k$& $n_2$& $m_2$\\\hline
Gonorrhoea sexual network 1~\cite{DeSinWon-STI-04}& 38& 39& 2& 2& 9& 10\\
Gonorrhoea sexual network 2~\cite{DeSinWon-STI-04}& 84& 90& 7& 4& 22& 28\\
Lippo Group Shipping~\cite{LomHob-03}& 96& 112& 17& 16& 45& 61\\ 
\begin{tabular}[x]{@{}l@{}}
Global International Airways\\
and Indian Springs State Bank~\cite{LomHob-03}
\end{tabular}
     & 82& 99& 18& 15& 33& 50\\ 
Gondola Genealogy~~\cite{NooMrvBat-12}& 242& 255& 14& 14& 50& 63\\
HIV~\cite{PotPhiPlu-STI-02}& 243& 257& 15& 12& 39& 53\\
Power Grid~\cite{WatStr-Nat-98}& 4941& 6594& 1654& 1516& 3353& 5006\\
\end{tabular}
\medskip
\caption{Statistics for real-world almost-trees.}
\label{tab:stats}
\end{table}
\vspace{-0.2in}

\section{The kernel}\label{sec:prelim} \label{sec:kernel}
Our fixed-parameter algorithms use the \emph{kernelization method}. In this
method we find a polynomial time transformation from an arbitrary input instance to a \emph{kernel}, an instance whose size is bounded by a fixed function $f(k)$ of the parameter value, and then apply a non-polynomial algorithm to the kernel. In this section we outline the general method for kernelization that we use in
our fixed parameter algorithms, based on a similar kernelization by Bannister, Cabello and Eppstein~\cite{BanCabEpp-WADS-2013} for a different problem, $1$-planarity testing.

We first describe our kernelization for cyclomatic number, which starts by
reducing the graph to its $2$-core.  The $2$-core of a graph can be found in linear time by initializing a queue of degree one vertices, repeatedly finding and removing vertices from the queue and the graph, and updating the degree and queue membership of the neighbor of each removed vertex.The $2$-core consists of vertices of degree at least three connected to each other by paths of degree two vertices. The following lemma bounds the numbers of high degree vertices and maximal paths of degree two vertices.

\begin{lemma}
    \label{lem:2core-bound}
If $G$ is a graph with cyclomatic number $k$ and minimum degree three
then $G$ has at most $2k - 2$ vertices and at most $3k-3$ edges. Furthermore, this bound is tight. As a consequence, the $2$-core of a graph with cyclomatic number $k$ has at most $2k-2$ vertices of degree at least three, and at most $3k-3$ maximal paths of degree two vertices.
\end{lemma}

\begin{proof}
Double counting yields $2(n-1+k) \geq 3n$, simplifying to $n \leq 2k-2$. A spanning tree of $G$ has at most $2k-3$ edges, and there are $k$ edges outside the tree, from which the bound on edges follows. For a graph realizing the upper bound consider any cubic graph with $2k-2$ vertices.
\end{proof}

The final step in this kernelization is to reduce the length of the maximal
degree two paths. 
Depending on the specific problem, we will determine a maximal allowed path length $\ell(k)$, and if any paths exceed this length we will shorten them to length exactly $\ell(k)$.
After this step the kernel will have $O(k\ell(k))$ edges and vertices, bounded by a function of $k$.

To change the parameter of our algorithms from the cyclomatic number to the $k$-almost-tree parameter, we first decompose the graph into its biconnected components. These components have a tree structure and in most drawing styles they can be embedded separately without introducing crossings. We then kernelize and optimally embed each biconnected component individually.

\section{$1$-page crossing minimization}
Minimizing crossings in $1$-page drawings is important for several drawing styles, but is \textsf{NP}-hard~\cite{MasKasNakFuj-ISCS-87}, leading Baur and Brandes to develop fast practical heuristics for reducing but not optimizing the number of crossings~\cite{BauBra-GTCCS-05}.
As we now show, crossing minimization and crossed edge minimization in $1$-page drawings of $k$-almost-trees is fixed-parameter tractable in the parameter $k$. We use the kernelization of Section~\ref{sec:kernel}, keeping one vertex per maximal degree two path.  

\begin{lemma}\label{lem:1-page-kernel}
Let $G$ have cyclomatic number $k$ and let $K$ be the kernel constructed from
$G$ with $\ell(k)=2$. Then
\begin{enumerate}
\item
$K$ has at most $5k$ vertices and $6k$ edges;
\item
$\pagecross_1(G) = \pagecross_1(K)$;
\item
$\pagecrossed_1(G) = \pagecrossed_1(K)$.
\end{enumerate}
\end{lemma}

\begin{proof}
(1) After reducing a graph with cyclomatic number $k$ to its $2$-core and reducing all maximal degree two paths to single edges we have a graph with $2k-2$ vertices and $3k-3$ edges, by Lemma~\ref{lem:2core-bound}. Since we are then adding one vertex back to every path that was not a single edge in the original graph, $K$ has at most $5k-5 \leq 5k$ vertices and $6k-6 \leq 6k$ edges.

(2) First we show that $\pagecross_1(G) \leq \pagecross_1(K)$. Suppose that $K$ has been embedded in one page with the minimum number of crossings. Every degree two vertex $v$ in $K$ corresponds to a path of degree two vertices in $G$. We can expand this path in a small neighborhood of $v$ without introducing any new crossings. After expanding all degree two paths we have an embedding of the $2$-core of $G$. Now each of the remaining vertices corresponds to a tree in $G$. Since trees can be embedded in one page without crossings, we can expand each tree in a small neighborhood of its corresponding vertex without introducing further crossings.

Now we show that $\pagecross_1(K) \leq \pagecross_1(G)$. Suppose that $G$ is embedded on one page with minimum crossings. Reduce $G$ to its $2$-core; this does not increase crossings. Let $u$ and $v$ be two adjacent degree-two vertices of $G$, let $e$ be the edge between $u$ and $v$ and let $f$ be the edge from $v$ not equal to $e$. Now, change the embedding of $G$ by keeping $u$ fixed and moving $v$ next to $u$, rerouting $f$ along the former path used by both $e$ and $f$. This change moves all crossings from $e$ to $f$ but does not create new crossings, so it produces another minimum-crossing embedding. After this change, edge $e$ may be contracted, again without changing the crossing number. Repeatedly moving one of two adjacent degree-two vertices and then contracting their connecting edge eventually produces an embedding of $K$ whose crossing number equals that of $G$.

(3) Follows from the proof of (2) with minor modification.
\end{proof}

\begin{lemma}
If $G$ is a graph with $n$ vertices and $m$ edges, then $\pagecross_1(G)$ and $\pagecrossed_1(G)$ can be computed in $O(n!)$ time.
\end{lemma}
\begin{proof}
We place the vertices in an arbitrary order on a circle, and compute the number of crossings or crossed edges for this layout. Then we use the Steinhaus--Johnson--Trotter algorithm~\cite{Sed-ACS-77} to list the $(n-1)!$ permutations of all but one vertex efficiently, with consecutive permutations differing by a transposition. When a transposition swaps $u$ and $v$, the number of crossings (or crossed edges) in the new layout can be updated from its previous value in $O(\deg(u) + \deg(v)) = O(n)$ time, as in~\cite{BauBra-GTCCS-05}. This yields a total run time of $O(n!)$.
\end{proof}

Combining the above lemmas, we apply the non-polynomial time algorithm only on
the kernel of the graph to achieve the following fixed parameter result.
\begin{theorem} \label{thm:1page}
If $G$ is a graph with cyclomatic number $k$, then $\pagecross_1(G)$ and $\pagecrossed_1(G)$ can be computed in $O((5k)! + n)$ time. If $G$ is a $k$-almost-tree, then $\pagecross_1(G)$ and $\pagecrossed_1(G)$ can be computed in $O((5k)!n)$ time.
\end{theorem}

In Section~\ref{sec:matmult} we show how to improve the base of the factorial in this bound by applying fast matrix multiplication algorithms.

\section{$2$-page crossing minimization}
In this section we consider the problem of $2$-page crossing minimization. I.e.,
we seek a circular arrangement of the vertices of a graph $G$, and an assignment
of the edges to either the interior or exterior of the circle, such that the
total number of crossings is minimized.
As in the $1$-page case, we consider minimizing both the number of
crossings  $\pagecross_2(G)$ and the
number of crossed edges $\pagecrossed_2(G)$.

There are two sources of combinatorial complexity for this problem, the vertex ordering and the edge assignment. However, even when the vertex ordering is fixed, choosing an edge assignment to minimize crossings is $\NP$-hard~\cite{mnkf-cmleg-90}. The hard instances of this problem can be chosen to be perfect matchings (with $k$-almost-tree parameter zero), so unless $\P=\NP$ there can be no FPT algorithm for the version of the problem with a fixed vertex ordering. Paradoxically, we show that requiring the algorithm to choose the ordering as well as the edge assignment makes the problem easier. A straightforward exact algorithm considers all $2^m(n-1)!$ possible configurations and chooses the one minimizing the total number of crossings, running in $O(2^mn!)$ time. We will combine this fact with our kernelization to produce an FPT algorithm.

We will give a sequence of reduction rules that transforms any drawing of $G$ into a drawing with the same number of crossings and crossed edges, in which the lengths of all paths are bounded by a function $f(k)$ of the parameter $k$. These reductions will justify the correctness of our kernelization using the same function.
Our reductions are based on the observation that, if $uv$ is an uncrossed edge, and $u$ and $v$ are consecutive vertices on the spine, then edge $uv$ can be contracted without changing the crossing number or number of uncrossed edges.
A given layout may not have any uncrossed edges connecting consecutive vertices, but we will show that, for a graph with a long degree two path, the layout can be modified to produce edges of this type without changing its crossings.

\begin{lemma}\label{lem:crossed-bound}
Let $G$ be a graph with cyclomatic number $k$. 
Then there exists a 2-page drawing with at most $k$ crossed edges, and at most 
${ k \choose 2}$ crossings.
\end{lemma}
\begin{proof}
Remove $k$ edges from $G$ to produce a forest, $F$.
Draw $F$ without crossings on one page, and draw the remaining $k$ edges on the other page.
Only the $k$ edges in the second page may participate in a crossing.
\end{proof}

We classify the possible configurations of pairs of consecutive edges of a degree two
path, up to horizontal and vertical symmetries, into four possible types: \emph{m, s, rainbow, and spiral}, as depicted in
Figure~\ref{fig:path-types}.
\begin{figure}[ht]
    \centering
    \includegraphics[width=.95\textwidth]{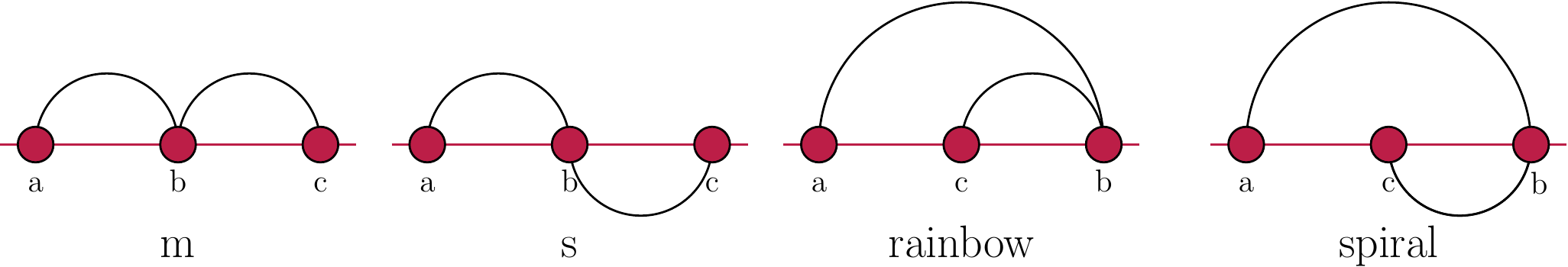} 
    \caption{Up to horizontal and vertical symmetry, the only possible
        arrangements of two consecutive edges are m, s, rainbow, and spiral.
        \label{fig:path-types}
    }
\end{figure}

\begin{lemma}\label{lem:mr}
If a layout contains a pair of edges $ab$ and $bc$ of m or rainbow type with edge $bc$ uncrossed and with $b$ and $c$ both having degree two, then it can be reduced without changing its crossings by a rearrangement followed by a contraction of the edge $bc$.
\end{lemma}

\begin{proof}
In either configuration we move vertex $b$ adjacent to vertex $c$, on the opposite side of $c$ from its other neighbor, as demonstrated in Figure~\ref{fig:mr}. Since the edge $bc$ is uncrossed this transformation does not change the crossing structure of the drawing. Now that $b$ and $c$ are placed next to each other the edge $bc$ may be contracted.
\end{proof}

\begin{figure}[ht]
\centering
\includegraphics[width=0.465\textwidth]{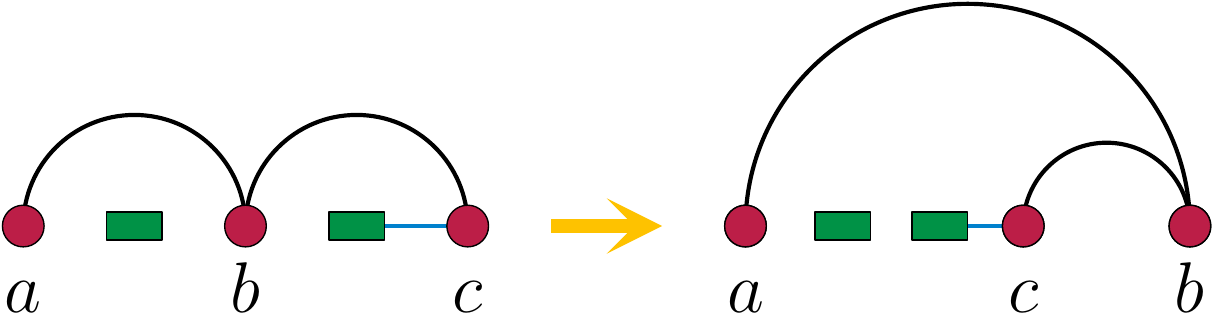}\hfill
\includegraphics[width=0.465\textwidth]{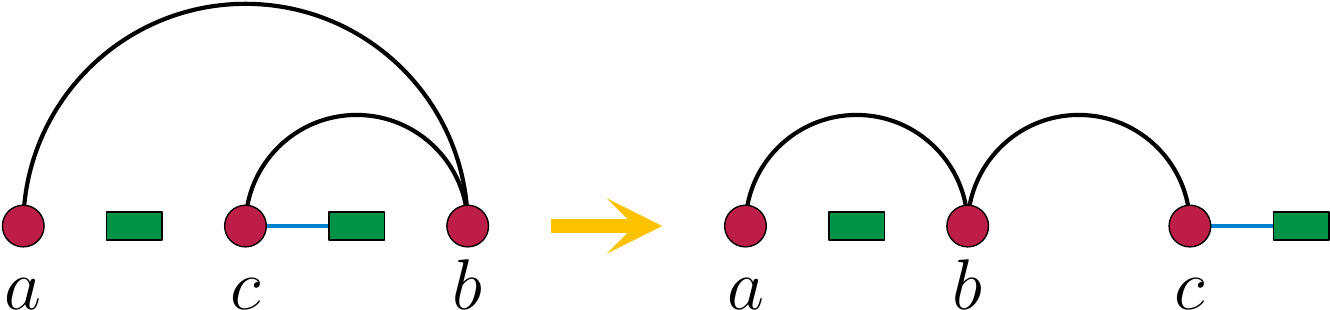}
\caption{The m reduction (left) and the rainbow reduction (right) shown with an edge into the $\beta$ region.}\label{fig:mr}
\end{figure}

\begin{figure}[ht]
\centering
\includegraphics[width=0.465\textwidth]{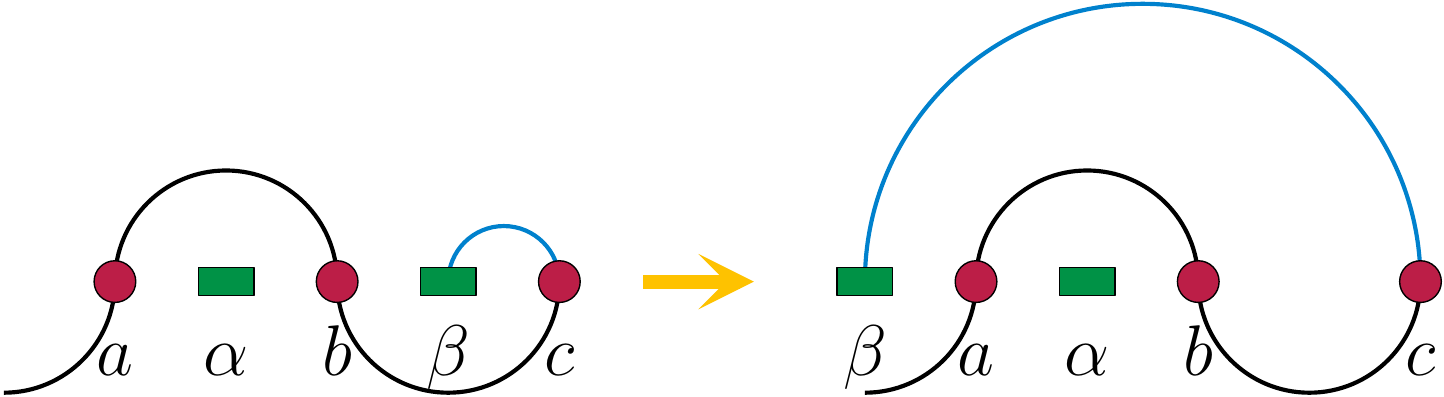}\hfill
\includegraphics[width=0.465\textwidth]{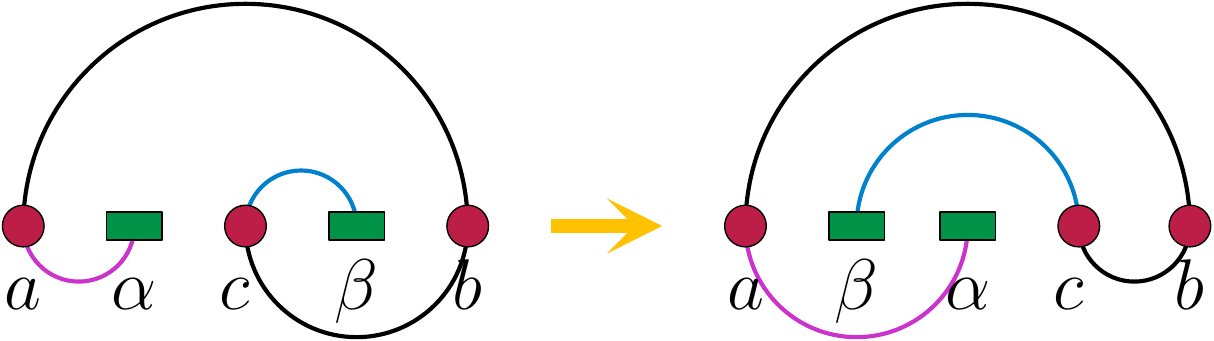}
\caption{The s reduction (left) shown with an edge into $\beta$ and the spiral reduction (right) shown with edges into $\alpha$ and $\beta$.}
\label{fig:s}
\end{figure}

\begin{lemma}
If a layout contains a pair of uncrossed edges $ab$ and $bc$ of s  or spiral type,  with  $a$, $b$, and $c$ all having degree two, then the layout can be reduced without changing its crossings by a rearrangement and contraction.
\end{lemma}

\begin{proof}
We assume by symmetry that $a$ is the leftmost of the three vertices, edge $ab$ is in the upper page, and edge $bc$ is in the lower page.
Let $x$ be the neighbor of $a$ that is not $b$ and let $y$ be the neighbor of $c$ that is not $b$. 
We may assume that edge $xa$ is in the lower page, for if it were in the upper page then  edge $ab$ would be part of an m or rainbow configuration and could be reduced by Lemma~\ref{lem:mr}. By the same reasoning we may assume that $cy$ is in the upper page.

First we consider s configurations.  Let $\alpha$ be the set of vertices between $a$ and $b$, and let $\beta$ be the set of vertices between $b$ and $c$. Then $\beta$ can have no incoming edges in the lower page, because $cy$ is upper and $bc$ blocks all other edges. Therefore, we may move $\beta$ directly to the left of $a$, as in Figure~\ref{fig:s}. Since edges $ab$ and $bc$ are uncrossed this transformation does not change the crossing structure of the drawing. We can then contract edge $ab$.

For the spiral, assume by symmetry that $c$ is between $a$ and $b$.
Let $\alpha$ be the set of vertices between $a$ and $c$, and let $\beta$ be the set of vertices between $c$ and $b$. Because $cy$ is assumed to be in the upper page, and $bc$ blocks all other lower edges, $\beta$ can have no incoming lower edges; however, it might have edges in the upper page connecting it to $\alpha$, so we must be careful to avoid twisting those connections and introducing new crossings. In this case, we move $\beta$ between $a$ and $\alpha$ and contract edge $bc$.
\end{proof}

As shown above, if any degree two path has at least four edges and two
consecutive uncrossed edges, then we can apply one of the reduction rules and
reduce the number of edges. For this reason we define the kernel $K$ for
computing $\pagecross_2(G)$ using the method in Section~\ref{sec:kernel}, with the bound $\ell(k)=2k^2$ on the length of the maximal degree two paths.
Similarly, we define the the kernel $L$ for computing $\pagecrossed_2(G)$ by setting $\ell(k)=2k$.

\begin{lemma} \label{lem:2page-reduce}
    Let $G$ be a graph with cyclomatic number $k$. Then,
    \begin{enumerate}
        \item $K$ has at most $6k^3$ vertices and $6k^3$ edges;
        \item $\pagecross_2(G) = \pagecross_2(K)$;
        \item $L$ has at most $6k^2$ vertices and $6k^2$ edges;
        \item $\pagecrossed_2(G) = \pagecrossed_2(L)$
    \end{enumerate}
\end{lemma}
\begin{proof}
(1) Since we have at most $2k^2$ vertices per maximal degree two path, the total number of vertices is at most $2k^2(3k-3) + 2k-2 \leq 6k^3$. The number of edges is at most $2k^2(3k-3) + (2k-2) + (k-1) \leq 6k^3$. 

(2) The proof that $\pagecross_2(G) \leq \pagecross_2(K)$ is that same as in Lemma~\ref{lem:1-page-kernel}. To see that $\pagecross_2(K) \leq \pagecross_2(G)$ we suppose that $G$ has been given an embedding that minimizes $\pagecross_2(G)$. The total number of crossings in such an embedding is bounded above by ${k \choose 2 }< k^2/2$, and in turn the number of crossed edges is less than $k^2$. Thus any maximal degree two path in $G$ with length greater than $2k^2$ can be shortened.

(3) and (4) The proof follows by the same argument as in (1) and (2), noting that there always exists a drawing with at most $k$ crossed edges.
\end{proof}

We apply the straightforward exact algorithm to the kernel of the graph to
achieve the following result:

\begin{theorem}\label{thm:2page}
If $G$ is a graph with cyclomatic number $k$, then $\pagecross_2(G)$
can be computed in $O(2^{6k^3}(6k^3)! + n)$ time, 
and $\pagecrossed_2(G)$ can be computed in $O(2^{6k^2}(6k^2)! + n)$ time. 
If $G$ is a $k$-almost-tree, then 
$\pagecross_2(G)$ and $\pagecrossed_2(G)$ can be computed in 
$O(2^{6k^3}(6k^3)!n)$
 time and 
$O(2^{6k^2}(6k^2)!n)$
time respectively.
\end{theorem}

\section{Matrix multiplication improvement}
\label{sec:matmult}

\begin{wrapfigure}[10]{r}{0.33\textwidth} 
\vspace{-3.78em}
\centering
\includegraphics[width=0.35\textwidth]{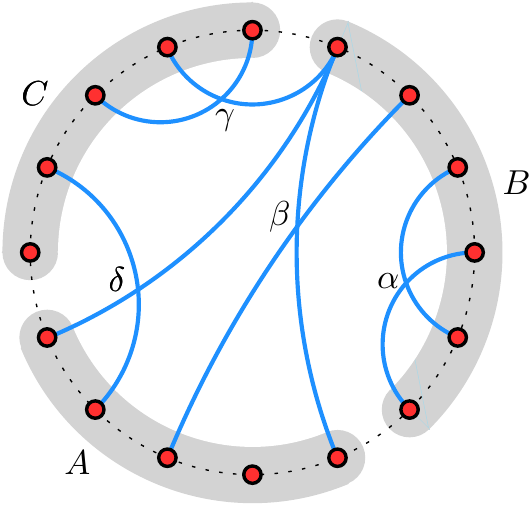}
\vspace{-2.1em}
\caption{Types of crossings.}
\label{fig:fpt-3part}
\end{wrapfigure}
The asymptotic run time for processing each biconnected component in both the
one page and two page cases can be further improved 
using matrix multiplication to find the minimum weight triangle in a graph\cite{Wil-TCS-05}. 

We begin with the 1-page case, in which we improve the run time to 
$O(k^{O(1)}(5k)!^{\omega/3})$ where $\omega<2.3727$ is the exponent for matrix multiplication~\cite{Wil-STOC-12}.
Let $N \leq 5k$ be the number of vertices in the kernel $K$, and for simplicity of exposition, 
assume that $N$ is a multiple of $3$. 
We construct a new graph $G'$ as follows. For each subset $S \subset K$ of $N/3$
vertices in the original kernel $K$, and for each ordering of $S$, we create one vertex in
$G'$. Thus, the number of vertices in $G'$ is 
\(
    (N/3)!\cdot{N \choose N/3} = O\left((N!)^{1/3}\right).
    \)
We add  edges in $G'$ between pairs of vertices that represent disjoint subsets.  $G'$ has a triangle for every triple of subsets that
form a proper partition of $V$ in $G$. Thus, each triangle corresponds to an
assignment of the vertices to three uniformly sized regions and a distinct ordering of the
vertices in each region, which together form a complete layout of $G$.

We assign a weight to each edge in $G'$ based on the
number of edge crossings in $G$ between the vertices in the corresponding
regions.
There are four possible types of crossing, represented by $\alpha$, $\beta$,
$\gamma$ and $\delta$ in Figure~\ref{fig:fpt-3part}. For a crossing of type
$\alpha$, in which all endpoints of a pair of crossing edges in $G$ are contained in the
same region $B$,
we add $1/2$ to the weights of edges $AB$ and $BC$ in $G'$. 
For $\beta$, in which a pair of crossing edges in $G$ both start in a region $A$ and
end in another region $B$, we add $1$ to the weight of edge $AB$ in $G'$. 
For $\gamma$, in which three endpoints of a pair of edges lie in the same
region $C$, and the fourth lies in a different region $B$, we add $1$ to the
weight of edge $BC$ in $G'$. 
Finally, for $\delta$, in which a pair of crossed edges both have an endpoint in
one region $A$, but their other endpoint in two different regions $B$ and $C$, 
we add $1/2$ to the weight of edge $AC$ and $1/2$ to the weight of edge $AB$ in
$G'$. With these weights, the total weight of a 
triangle in $G'$ equals the number of edge crossings in the corresponding layout. 
The edge weights for $G'$ can be computed in $O(k^{O(1)} (5k)!^{2/3})$ time.

To find the minimum weight triangle we construct the weighted adjacency matrix $A$, where $A_{i,j}$ is given the weight of the edge from $i$ to $j$ or infinity if no such edge exists. We then compute the min-plus matrix product of $A$ with itself, which is defined by $[A \star A]_{i,j} = \min_k A_{i,k} + A_{k,j}$. The weight of a minimum weight triangle in $A$ then corresponds to the minimum entry in $A + A \star A$. From the minimum weight and corresponding $i$ and $j$ the triangle can be found in linear time. Thus, the runtime is dominated by computing $A \star A$, which can be done in $O(k^{O(1)}(5k)!^{\omega/3})$ time using fast matrix multiplication~\cite{AloGalMar-JCSS-97, Yuv-IPL-76}.

For the $2$-page case we consider each of the $2^M$ edge page assignments separately, computing the minimum crossing drawing for this assignment using matrix multiplication. As before we construct a graph $G'$ with weighted edges between compatible vertices, such that a minimum weight triangle in $G'$ corresponds to a minimum weight drawing. Matrix multiplication is then used to find this minimal weight triangle for each page assignment, yielding a running time of $O(2^M(N!)^{\omega/3})$, where $N$ is the number of vertices and $M$ is the number of edges in the kernel. Thus, we have the following
result: 

 \begin{theorem}\label{thm:matrix}
If $G$ is a graph with cyclomatic number $k$, then we can compute:
\begin{itemize}
    \item $\pagecross_1(G)$ and $\pagecrossed_1(G)$ in $O(k^{O(1)}(5k)!^{\omega/3} + n)$ time;
    \item $\pagecross_2(G)$ in $O(2^{6k^3}(6k^3)!^{\omega/3}
        + n)$ time;
    \item $\pagecrossed_2(G)$ in $O(2^{6k^2}(6k^2)!^{\omega/3} + n)$
        time.
\end{itemize}
\end{theorem}

\section{Conclusion}
We have given new fixed parameter algorithms for computing the minimum number of
edge crossings and minimum number of crossed edges in $1$-page and $2$-page
embeddings of $k$-almost trees. To our knowledge, these are the only
parameterized exact algorithms for these drawing styles. 

We leave the following questions open to future research:
\begin{itemize}
\item For $2$-page embeddings, the hardness of finding uncrossed drawings~\cite{ChuLeiRos-SJADM-87} shows that crossing minimization cannot be FPT in its natural parameter, the number of crossings. What about $1$-page embeddings?
\item Can the dependence on $k$ be reduced to singly exponential?
\item What other $\NP$-hard problems in graph drawing are FPT with respect to $k$?
\end{itemize}

\subsubsection{Acknowledgements.}
We thank Emma S. Spiro whose work with social networks led us to consider drawing almost-trees, and an anonymous reviewer for helpful suggestions in our 2-page kernelization. This research was supported in part by the National Science Foundation under grant 0830403 and 1217322, and by the Office of Naval Research under MURI grant N00014-08-1-1015.

{\raggedright
\bibliographystyle{abuser}
\bibliography{drawing-almost-trees}}

\end{document}